\documentclass[letterpaper, 10 pt, conference]{ieeeconf}  

\IEEEoverridecommandlockouts                              
\usepackage{amsmath, amssymb, amsfonts,bm}
\usepackage{mathrsfs}
\usepackage[mathcal]{euscript}
\usepackage{enumerate}
\usepackage{mathbbol}

\usepackage{float}

\usepackage{xcolor,graphicx,epstopdf}
\usepackage{dsfont}
\usepackage{hyperref}
\usepackage{caption}
\usepackage{subcaption}
\usepackage{cite} 

\usepackage{tikz-cd}
\usepackage{mathtools}

\usepackage[subfigure]{tocloft}
\usepackage{graphicx}
\usepackage{float}
\usepackage{epsfig}

\usepackage[ruled,vlined]{algorithm2e}

\usepackage{tikz}
\usepackage{pgfplots}
\usetikzlibrary{positioning,calc}

\usepackage{xcolor}
\definecolor{midnightblue}{rgb}{0, 0.14,0.55}
\usepackage[normalem]{ulem}

\newcommand*\mcapinn[2]{\vcenter{\hbox{$\mathsurround=0pt
			\ifx\displaystyle#1\textstyle\else#1\fi\bigcap$}}}

\newcommand*\mcupinn[2]{\vcenter{\hbox{$\mathsurround=0pt
			\ifx\displaystyle#1\textstyle\else#1\fi\bigcup$}}}

\newtheorem{theorem}{Theorem}
\newtheorem{definition}{Definition}

\newtheorem{remark}{Remark}
\newtheorem{proposition}{Proposition}
\newtheorem{assumption}{Assumption}
\newtheorem{example}{Example}

\title{\LARGE \bf
Competitive Equilibrium in Microgrids With Dynamic  Loads
}

\author{Zeinab Salehi, Yijun Chen, Ian R. Petersen, Elizabeth L. Ratnam, and Guodong Shi
\thanks{This work was supported by the Australian Research Council under grants DP190102158, DP190103615, LP210200473, and DP230101014.}
\thanks{Z. Salehi, I. R. Petersen and  E. L. Ratnam are with the Research School of Engineering, The Australian National University, Canberra, Australia. (E-mail:  zeinab.salehi@anu.edu.au; ian.petersen@anu.edu.au; elizabeth.ratnam@anu.edu.au) }
\thanks{Y. Chen and G. Shi are with the  Australian Center for Field Robotics, School of Aerospace, Mechanical and Mechatronic Engineering, The University of Sydney, NSW, Australia. (e-mail: yijun.chen@sydney.edu.au; guodong.shi@sydney.edu.au)}
}

\begin{document}
	
\maketitle
\thispagestyle{empty}
\pagestyle{empty}

\begin{abstract}
	In this paper, we consider microgrids that interconnect prosumers with distributed energy resources and dynamic loads. Prosumers are connected through the microgrid to trade energy and gain profit while respecting the network constraints. We establish a local energy market by defining a  competitive equilibrium which balances energy and satisfies voltage constraints within the microgrid for all time. Using duality theory, we prove that under some convexity assumptions, a competitive equilibrium is equivalent to a social welfare maximization solution. Additionally, we show that a competitive equilibrium is equivalent to a Nash equilibrium of a standard game. In general, the energy price for each prosumer is different, leading to the concept of \textit{locational prices}.  We investigate a case under which  all prosumers have the same locational prices. Additionally, we show that under some assumptions on the resource supply and network topology, locational prices decay to zero after a period of time, implying the available supply will be more than the demand required to stabilize the system. Finally, two numerical examples are provided to validate the results, one of which is a direct application of our results on electric vehicle charging control.
\end{abstract}

\section{Introduction}
A microgrid is an electrical network that connects a group of loads and distributed energy resources (DERs) such as solar panels,  battery storage, and electric vehicles (EVs) --- which can operate in isolation from the bulk grid \cite{sandelic2022reliability}. Market-based approaches can encourage end customers owning DERs (prosumers) to participate in microgrid operations while maximizing their economic profits and utilities \cite{zhou2017incentive}. Prosumers that participate in microgrid operations can improve voltage regulation \cite{AmritPaudel2021},  frequency regulation \cite{King2021Solving}, and  demand-side management \cite{tao2021data}.

To incentivize local market participation within a microgrid we consider a competitive equilibrium \cite{Li2020}. The concept of a competitive equilibrium was first introduced in microeconomics as the pair of resource price and consumed resource which clears the market \cite{mas1995microeconomic}. At a competitive equilibrium, each prosumer maximizes its individual payoff, as the summation of the income from resource trading and the utility from resource consumption,  such that the total supply balances the total demand. A competitive equilibrium is analogous to frequency regulation in a power system \cite{salehi2021social}. 
While a competitive equilibrium maximizes the payoff of individuals considering their local constraints, a social welfare problem maximizes the utility for an entire community considering both local and global constraints \cite{chen2022competitive}. 

A competitive equilibrium is  Pareto optimal, in the sense that no one can change their decision unilaterally without reducing the payoff of at least one other individual \cite{acemoglu2018}. It is proved that under some convexity assumptions, a competitive equilibrium and a  social welfare maximization solution are equivalent \cite{Li2020, mas1995microeconomic}. The relationship between a competitive equilibrium and a generalized Nash equilibrium has been widely studied in the literature \cite{facchinei2010generalized, herves2020market}. Arrow and Debreu \cite{arrow1954existence} were the first to show that a competitive equilibrium can be obtained as a solution to a generalized game with a price player.

Several authors have considered a competitive equilibrium in markets with static loads \cite{Li2020, salehi2021quadratic, salehi2021social}, while others have focused on dynamic loads \cite{chen2022competitive, salehi2022competitive, salehi2022infinite, salehi2023competitive}. 
Dynamic loads such as EVs and thermostatically controlled loads (TCLs) can be operated to improve the resilience of a microgrid\cite{papari2021metrics}.
To ensure the power delivered to microgrid customers is of a high quality, voltages must also be regulated \cite{kang2022event}. Accordingly, several authors have considered market mechanisms focusing on voltage regulation in power systems \cite{liu2022fully, umer2023novel}. For example, Li in \cite{li2015market} introduces a competitive equilibrium to balance the energy in a microgrid subject to voltage constraints and static loads.

In this paper, we extend the work in \cite{li2015market} to dynamic loads by defining a new competitive equilibrium that incorporates grid constraints. We consider locational prices associated with a competitive equilibrium which consist of two parts: (i) the energy price reflecting the energy balance across the network; and (ii) prices reflecting voltage regulation within the microgrid. Depending on the topology of the microgrid, the locational price at each node can be different. The main contributions of the paper are as follows.
\begin{itemize}
\item Using duality theory, we prove that under some convexity assumptions, a  competitive equilibrium is equivalent to a social welfare maximization solution. 

\item  We show that a competitive equilibrium is equivalent to a Nash equilibrium of a standard game.
\item We present conditions under which the locational prices are the same across the network. 

\item We show that under certain assumptions on the resource supply and network topology, the locational prices have a decaying behavior.
\end{itemize}
 
The rest of the paper is organized as follows. In Section \ref{sec:preliminaries}, we describe a microgrid with dynamic loads, and the associated power flow equations. Additionally, we present a real-world application (EV charging) to motivate our problem formulation. In Section \ref{sec:main}, we define a competitive equilibrium and a social welfare maximization problem, followed by investigating some of their properties. Then, we consider the relationship between a competitive equilibrium and a Nash equilibrium. Finally, Section \ref{sec:examples} provides  simulation results (including EV charging), and  Section \ref{sec:conclusions} concludes the paper. 
\section{Problem Formulation} \label{sec:preliminaries}
In this section, we describe a microgrid that serves dynamic loads. Then, we describe EV charging as a real-world application.
\subsection{Microgrids With Dynamic Loads}
Consider a radial microgrid with $n$ prosumers indexed in the set  $\mathcal{N}=\{1, 2, ..., n \}$. Each prosumer $i \in \mathcal{N}$  corresponds to a residential, commercial, or an industrial building with  uncontrollable loads such as  lights,  and  controllable loads such as EVs and TCLs. Each prosumer is also equipped with local energy production such as solar panels. We study the microgrid over a finite time horizon which is divided into $T$ time intervals of length $\Delta$.   Time intervals $t$ are indexed in the set $\mathcal{T} = \{0, \cdots, T-1\}$. The net supply (i.e., generation minus uncontrollable loads) for prosumer $i \in \mathcal{N}$ is denoted by $a_i(t) \in \mathbb{R}$. The dynamics associated with the controllable loads of prosumer $i \in \mathcal{N}$ are described by
\begin{equation}\label{eq_state}
	\mathbf x_i(t+1)= \mathbf A_i \mathbf x_i(t)+ \mathbf B_i \mathbf u_i(t), \quad t \in \mathcal{T},
\end{equation}
where $\mathbf x_i(t) \in \mathbb{R}^d$ is the state (e.g.,  state of charge (SoC) of an EV battery) and  $\mathbf u_i(t) \in \mathbb{R}^m$ is the control input (e.g.,  charge and discharge rate of an EV). The states and control inputs are physically constrained by
\begin{equation*}
	\underline{\mathbf x}_i \leq \mathbf x_i(t) \leq \overline{\mathbf x}_i, 
	\quad  \underline{\mathbf u}_i\leq \mathbf u_i(t) \leq \overline{\mathbf u}_i, 
\end{equation*}
where $\underline{\mathbf x}_i $ and $\underline{\mathbf u}_i$ are the lower bounds, and $\overline{\mathbf x}_i$ and $\overline{\mathbf u}_i$  are the upper bounds. Reaching the state $\mathbf x_i(t)$ and applying the control input $\mathbf u_i(t)$, prosumer $i$ achieves a running utility $f_i(\mathbf x_i(t), \mathbf u_i(t)) : \mathbb{R}^{d} \times \mathbb{R}^{m} \mapsto \mathbb{R}$ at time $t \in \mathcal{T}$, and a terminal utility $\Phi_i(\mathbf x_i(T)) : \mathbb{R}^{d} \mapsto \mathbb{R}$ at the final time step. The energy consumption/withdrawal as a result of taking the control action $\mathbf u_i(t)$ is denoted by $h_i(\mathbf u_i(t)): \mathbb{R}^{m} \mapsto \mathbb{R}$.
Each prosumer $i$ has a surplus/shortage of energy $a_i(t) - h_i(\mathbf u_i(t))$ which can be traded through the microgrid. In what follows, we describe power flows in a microgrid and the associated voltage constraints that must be respected.
\subsection{Microgrid Power Flows}
Each prosumer is connected to a node in the microgrid. The  network has $n+1$ nodes indexed in the set $\mathcal{N}\cup \{0\}$ where $0$ is the reference node (feeder). The set of all lines in the network is denoted by $\mathcal{E}$, where  $(i, j) \in \mathcal{E}$ represents the line connecting nodes $i$ and $j$ with the resistance  $r_{ij}$ and reactance $x_{ij}$. Denote by $V_i(t) \in \mathbb{C}$  the voltage phasor at node $i \in \mathcal{N}\cup \{0\}$ and time step $t$, and $v_i(t)=|V_i(t)|^2$. The square magnitude of the reference voltage is given and fixed, i.e., $v_0(t)=v_0$. Let $p_i(t) \in \mathbb{R}$ and $q_i(t) \in \mathbb{R}$ denote the  net average active and reactive power injections at node $i \in \mathcal{N}\cup \{0\}$ and time interval $t \in \mathcal{T}$, respectively. We suppose there is no external power injection to the network, i.e., $p_0(t)=q_0(t)=0$. The average active and reactive power flows from node $i$ to $j$ are denoted by $P_{ij}(t)$ and $Q_{ij}(t)$, respectively. Power flow  in the network can be described by the LinDistFlow model \cite{baran1989optimal}
\begin{equation}\label{eq1}
	\begin{aligned}
		&P_{ij}(t) = -p_j(t) + \sum_{k:(j,k) \in \mathcal{E}} P_{jk}(t),\\
		&Q_{ij}(t) = -q_j(t) + \sum_{k:(j,k) \in \mathcal{E}} Q_{jk}(t),\\
		&v_i (t)- v_j(t)= 2\big(r_{ij}P_{ij}(t) + x_{ij}Q_{ij}(t) \big).
	\end{aligned}
\end{equation}
According to \eqref{eq1}, the injected power must be balanced; that is,
\begin{equation}\label{eq_balance}
	\sum_{i=1}^{n} p_i(t)=0.
\end{equation}
Let $\mathcal{P}_i$ denote the unique path from node $0$ to node $i$. The voltage equation in \eqref{eq1} can be simplified as in \cite{farivar2013equilibrium} to
\begin{equation}\label{eq3}
	v_i (t)=  v_0 +\sum_{k=1}^{n} \big(R_{ik} p_k (t)+  X_{ik}  q_k (t)\big), \quad i \in \mathcal{N},
\end{equation}
where we have 
$R_{ik}=2\sum_{(h, l) \in \mathcal{P}_i \cap \mathcal{P}_k} r_{hl}$, and $X_{ik}=2\sum_{(h, l) \in \mathcal{P}_i \cap \mathcal{P}_k} x_{hl}$. At each node $i \in \mathcal{N}$, the voltage magnitude must remain within an acceptable range (typically  $\pm5 \% $ of a nominal voltage). For simplicity, we suppose  $q_i(t)$ is a given constant at each node $i \in \mathcal{N}$, i.e., $q_i(t)=q_i$. The voltage constraints can be written as 
\begin{equation}\label{eq_voltage}
	\underline {v}_i \leq \sum_{k =1}^n R_{ik}  p_k(t) \leq \overline { v}_i, \quad i \in \mathcal{N},
\end{equation}
where $\underline { v}_i$ and $\overline { v}_i$ are the lower and upper bounds on $v_i (t) -  v_0 -\sum_{k=1}^{n}   X_{ik}  q_k $. In the remainder, by power we mean average active power.

For node (prosumer) $i \in \mathcal{N}$, let  $\lambda_i(t) \in \mathbb{R}$ denote  the  \textit{locational price} for  unit  energy injection (energy trading) during time interval $t \in \mathcal{T}$.  Given $\lambda_i(t)$, prosumer $i$ decides about the amount of its  power injection $p_i(t)$ that is physically constrained by $p_i(t) \Delta\leq a_i(t) - h_i(\mathbf u_i(t))$; the prosumer aims to maximize  its total payoff as the summation of the utility from energy consumption and the income $\lambda_i(t)p_i(t)\Delta$ from energy trading, over the whole horizon.
\subsection{Motivating Example (EV Charging)}\label{sec:motivating_example}
As a real-world application, one may consider EV charging, where  each prosumer corresponds to a residential building with EVs as controllable loads. The dynamics are represented by \cite{nimalsiri2021coordinated}
\begin{equation}\label{eq_EV}
	x_i(t+1)= x_i(t)+ \eta_i u_i(t) \Delta,
\end{equation}
where $x_i(t)$ is the SoC indicating the energy remained in the battery (in kWh), $u_i(t)$ is the charge/discharge rate (in kW), and $\eta_i$ is the charge/discharge efficiency, for EV $i$ at time $t$. The consumed/withdrawn energy by EV $i$ is $h_i(u_i(t))=u_i(t)\Delta$, where $\Delta$ is the sampling time (in hour). Each EV has $C$  kWh battery capacity. To extend  the EV battery lifespan, the SoC should remain within $20\%$ to $85 \%$ of the battery capacity \cite{xing2015decentralized}, so $0.2C \leq x_i(t) \leq 0.85C$. Additionally, the charge/discharge rate is bounded by $-\overline {u} \leq u_i(t) \leq \overline {u}$, where $\overline {u}$  depends on the charger or the battery type. We study the network over a time horizon $T$.  All EVs arrive for charging/discharging at the beginning of the horizon and depart at the end of the horizon. Each prosumer $i \in \mathcal{N}$ tries to maximize its  payoff with the utility functions $f_i(x_i(t), u_i(t))=-u_i^2(t)$ and $\Phi_i(x_i(T))=-(x_i(T)-0.85C)^2$. We aim to design a local market which serves the interests of prosumers while respecting the network constraints \eqref{eq_balance} and \eqref{eq_voltage}.
\section{Main Results}\label{sec:main}
In this section, we propose a decentralized energy market that respects the grid constraints, in particular voltage constraints,  by defining a competitive equilibrium. 

\subsection{Decentralized Energy Market}
For prosumer $i \in \mathcal{N}$,  denote $\mathbf U_i=(\mathbf u_i^\top(0), \dots, \mathbf u_i^\top(T-1))^\top$ and $\mathbf p_i=( p_i(0), \dots,  p_i(T-1))^\top$  as the vectors of control inputs and  power injections over the time horizon, respectively. At time step $t \in \mathcal{T}$, Denote $\mathbf u(t)= (\mathbf u_1^\top(t), \dots, \mathbf u_n^\top(t) )^\top$, $\mathbf p(t)= ( p_1(t), \dots,  p_n(t) )^\top$, and  $\mathbf x(t)= (\mathbf x_1^\top(t), \dots, \mathbf x_n^\top(t) )^\top$ as the vectors of  control inputs, power injections, and dynamical states corresponding to all prosumers, respectively. We denote by $\mathbf U=(\mathbf u^\top(0), \dots, \mathbf u^\top(T-1))^\top$ and $\mathbf P=(\mathbf p^\top(0), \dots, \mathbf p^\top(T-1))^\top$  the vectors of control inputs and  power injections associated with all prosumers over the whole horizon, respectively.  Let $\boldsymbol  \lambda(t) = (\lambda_1(t), \dots, \lambda_{n}(t))^\top$ and $\boldsymbol{\Lambda}=(\boldsymbol{\lambda}^\top(0), \dots, \boldsymbol{\lambda}^\top(T-1))^\top$ denote the vectors of  energy prices associated with all prosumers at time step $t\in \mathcal{T}$ and over the whole time horizon, respectively.  
\begin{definition}\label{definition1}
	The triplet $(\boldsymbol\Lambda^\ast, \mathbf{U}^\ast, \mathbf{P}^\ast)$ is a \textit{competitive equilibrium} if the following conditions are satisfied.
	\begin{itemize}
		\item[(i)]  At the equilibrium, the individual payoff function of each prosumer $i\in \mathcal{N}$ is maximized subject to the local constraints; i.e., $(\mathbf U_i^\ast, \mathbf p_i^\ast)$ solves
		\begin{equation}\label{opt_LTD_1}
			\resizebox {0.447\textwidth} {!} {$
			\begin{aligned}
				\max_{{\mathbf U_i}, \mathbf p_i} \quad &  \sum_{t=0}^{T-1} \Big(f_i(\mathbf x_i(t), \mathbf u_i(t)) + \lambda^\ast_i(t) p_i(t)\Delta \Big)+{\Phi_i(\mathbf x_i(T)}) \\
				{\rm s.t.} 
				\quad &  \mathbf x_i(t+1)= \mathbf A_i \mathbf x_i(t)+ \mathbf B_i \mathbf u_i(t),  \\
				\quad &  p_i(t)\Delta \leq a_i (t)- h_i(\mathbf u_i(t)),  \\
				\quad & \underline{\mathbf x}_i \leq \mathbf x_i(t) \leq \overline{\mathbf x}_i, \\
				\quad & \underline{\mathbf u}_i\leq \mathbf u_i(t) \leq \overline{\mathbf u}_i,
				\quad  t \in \mathcal{T}.
			\end{aligned}$}
		\end{equation}
		\item[(ii)] At the equilibrium, the power flow is balanced; i.e.,
		\begin{equation}\label{constraint_balancing}
			\sum_{i=1}^n p_i^\ast(t) =0, \quad t\in \mathcal{T}.
		\end{equation} 
	 \item[(iii)] At the equilibrium, the voltage constraints are satisfied; i.e.,
	 \begin{equation}\label{equilibrium_voltage}
	 	\underline {v}_i \leq \sum_{k =1}^n R_{ik}  p_k^\ast(t) \leq \overline { v}_i, \,\, i \in \mathcal{N},  t \in \mathcal{T}.
	 \end{equation} 
	
		\item[(iv)] Considering $\boldsymbol\Lambda^\ast$, the locational prices satisfy
		\begin{equation}\label{constraint_price}
			\lambda_i^\ast(t)=\alpha^\ast(t)+\sum_{k=1}^{n} \big[\underline\xi_k^\ast(t)-\overline\xi_k^\ast (t) \big] R_{ki}, \,\, i \in \mathcal{N},  t \in \mathcal{T},
		\end{equation}
	where $\alpha^\ast(t)$ is the energy price and $\underline\xi_k^\ast(t)-\overline\xi_k^\ast (t)$ is the  price for a unit voltage change at node $k$.
		\item[(v)]   If the voltage constraints are not binding, the price for voltage change is zero; that is,
		\begin{equation}\label{constraint_voltage}
			\resizebox {0.44\textwidth} {!} {$
				\begin{aligned}
			\underline \xi_i^\ast(t)\Big(\underline v_i - \sum_{k =1}^n R_{ik}  p_k^\ast(t)\Big)=\overline \xi_i^\ast(t) \Big( \sum_{k =1}^n R_{ik}  p_k^\ast (t)-\overline{ v}_i \Big)=0,\\ \,\, i \in \mathcal{N},  t \in \mathcal{T}.
		\end{aligned}$}
		\end{equation}
	\end{itemize}
\end{definition}

Under the proposed competitive equilibrium, a local energy market can be established which respects the voltage constraints and balances the power flow. The optimal price $\lambda^\ast_i(t)$ is determined by the grid aggregator while the optimization problem in \eqref{opt_LTD_1} is solved independently by each prosumer who acts as a price taker. Next, we define the notion of social welfare.

\begin{definition}
	The pair $(\mathbf{U}^\star, \mathbf{P}^\star)$ maximizes the \textit{social welfare} if it maximizes the total utilities of all prosumers over the whole horizon subject to the local and global network constraints. In other words, $(\mathbf{U}^\star, \mathbf{P}^\star)$ is an optimizer to
		\begin{equation}\label{eq4}
		\begin{aligned}
			\max_{{\mathbf U}, \mathbf P} \quad &  \sum_{i=1}^{n}\Big(\sum_{t=0}^{T-1} f_i(\mathbf x_i(t), \mathbf u_i(t))+ {\Phi_i(\mathbf x_i(T))} \Big)\\
			{\rm s.t.} 
			\quad &  \mathbf x_i(t+1)= \mathbf A_i \mathbf x_i(t)+ \mathbf B_i \mathbf u_i(t),  \\
			\quad &  p_i(t)\Delta \leq a_i (t)- h_i(\mathbf u_i(t)),  \\
			\quad &	\sum_{i=1}^n p_i(t) = 0,\\
			\quad &\underline {v}_i \leq \sum_{k =1}^n R_{ik}  p_k(t) \leq \overline { v}_i, \\
			\quad & \underline{\mathbf x}_i \leq \mathbf x_i(t) \leq \overline{\mathbf x}_i, \\
			\quad & \underline{\mathbf u}_i\leq \mathbf u_i(t) \leq \overline{\mathbf u}_i,   \quad  i \in \mathcal{N}, \, t \in \mathcal{T}.
		\end{aligned}
	\end{equation}
\end{definition}
\medskip

Social welfare maximization is important from the system-level perspective, and it is solved by the aggregator.

\begin{remark}
	The competitive equilibrium  in Definition \ref{definition1}, which considers dynamic loads in a microgrid, is an extension of the works presented in \cite{li2015market}  (considering  static loads in a microgrid) and \cite{chen2022competitive} (considering dynamic loads without grid constraints).
\end{remark}
\subsection{Competitive Equilibrium \& Social Welfare}
 In the following, we show that under some conditions, a competitive equilibrium and a social welfare maximization solution coincide with each other. 
\begin{assumption}\label{assumption1}
	For $i \in \mathcal{N}$, let $f_i(\cdot)$ and $\Phi_i(\cdot)$ be concave functions, and $h_i(\cdot)$ be a convex function. Suppose Slater's condition holds for \eqref{opt_LTD_1} and \eqref{eq4} \cite{boyd2004convex}.
\end{assumption}

\begin{theorem}\label{theorem1}
Let Assumption \ref{assumption1} hold. Given a feasible initial condition $\mathbf x(0)$,  a competitive equilibrium and a social welfare maximization solution are equivalent, i.e., the following statements hold. 
\begin{itemize}
\item [(i)] If $( \mathbf U^\star, \mathbf P^\star)$ maximizes the social welfare, then there exists $\boldsymbol{\Lambda}^\ast \in \mathbb{R}^{nT}$ such that $(\boldsymbol{\Lambda}^\ast, \mathbf U^\star, \mathbf P^\star)$ is a competitive equilibrium.
\item [(ii)]  If $(\boldsymbol{\Lambda}^\ast, \mathbf U^\ast, \mathbf P^\ast)$ is a competitive equilibrium, then $(\mathbf U^\ast, \mathbf P^\ast)$ maximizes the social welfare.
\end{itemize}
\end{theorem}

\begin{proof}
(i) According to \eqref{eq_state}, the state $\mathbf x_i(t)$ can be written as a linear combination of the initial state $\mathbf x_i(0)$ and the control sequence $\mathbf U_i$ as
\begin{equation}\label{eq2}
	\mathbf x_i(t)=\mathbf A_i^t \mathbf x_i(0)+ \sum_{j=0}^{t-1}{\mathbf A_i^{t-j-1} \mathbf B_i \mathbf u_i(j)}, \quad t \in \{1, 2, ..., T\}.
\end{equation}
Substituting \eqref{eq2} into $f_i(\cdot)$ and $\Phi_i(\cdot)$, we have
\begin{equation}
	\begin{aligned}
		f_i(\mathbf x_i(t), \mathbf u_i(t)) &:= \tilde{f}_{i,t}(\mathbf U_i),\\
		\Phi_i(\mathbf x_i(T)) &:= \tilde{\Phi}_i(\mathbf U_i),
	\end{aligned}
\end{equation}
where $\tilde{f}_{i,t}(\cdot)$ and $\tilde{\Phi}_i(\cdot)$ are concave functions (the composition of a concave function and an affine function is concave). For $i \in \mathcal{N}$, define $\mathbb{U}_i =\{ \mathbf U_i| \underline{\mathbf u}_i\leq \mathbf u_i(t) \leq \overline{\mathbf u}_i;  \underline{\mathbf x}_i \leq \mathbf x_i(t) \leq \overline{\mathbf x}_i, \text{for } t \in \mathcal{T}\}$, which is a polyhedral set. For any $(\mathbf U, \mathbf P)$ such that $\mathbf U_i \in \mathbb{U}_i$, the Lagrangian function of \eqref{eq4} is defined as
\begin{equation}\label{eq5}
	\begin{aligned}
	L(\mathbf U, \mathbf P, &\boldsymbol{\alpha}, \underline{\boldsymbol\Xi}, \overline{\boldsymbol\Xi}, \boldsymbol\Psi) = - \sum_{i=1}^{n}\Big(\sum_{t=0}^{T-1} \tilde{f}_{i,t}( \mathbf U_i)+\tilde{\Phi}_i(\mathbf U_i) \Big) 
	\\&+ \sum_{t=0}^{T-1}\alpha(t) \Big(-\sum_{i=1}^n p_i(t) \Big) \Delta
	\\&+\sum_{t=0}^{T-1}\sum_{i=1}^{n}\underline\xi_{i}(t) \Big( \underline v_i - \sum_{k =1}^n R_{ik}  p_k(t) \Big) \Delta
	\\&+ \sum_{t=0}^{T-1}\sum_{i=1}^{n}\overline\xi_{i}(t) \Big( \sum_{k =1}^n R_{ik}  p_k(t)-\overline v_i \Big) \Delta
	\\&+ \sum_{t=0}^{T-1}\sum_{i=1}^{n} \psi_{i}(t) \Big( p_i(t) \Delta+h_i(\mathbf u_i(t))-a_i(t)\Big),
\end{aligned}
\end{equation}
where $\underline\xi_{i}(t), \overline\xi_{i}(t), \psi_{i}(t) \geq 0$, $\boldsymbol \alpha =(\alpha(0), \dots, \alpha({T-1}))^\top$, $ \underline{\boldsymbol\Xi}_i =(\underline\xi_{i}(0), \dots, \underline\xi_{i}(T-1))^\top$, and $ \underline{\boldsymbol\Xi} = ( \underline{\boldsymbol\Xi}_1^\top, \dots,  \underline{\boldsymbol\Xi}_n^\top)^\top$. We define $\overline{\boldsymbol\Xi}_i, \overline{\boldsymbol\Xi}$, and $\boldsymbol\Psi_i, \boldsymbol\Psi$ in a similar way. 
The Lagrangian  \eqref{eq5}  can be written as 
\begin{equation}\label{eq6}
	L(\mathbf U, \mathbf P, \boldsymbol{\alpha}, \underline{\boldsymbol\Xi}, \overline{\boldsymbol\Xi}, \boldsymbol\Psi) = \sum_{i=1}^{n} L_i(\mathbf U_i, \mathbf p_i, \boldsymbol{\alpha}, \underline{\boldsymbol\Xi}, \overline{\boldsymbol\Xi}, \boldsymbol\Psi_i),
\end{equation}
where
\begin{equation}
	\begin{aligned}
	L_i &(\mathbf U_i, \mathbf p_i, \boldsymbol{\alpha}, \underline{\boldsymbol\Xi}, \overline{\boldsymbol\Xi}, \boldsymbol\Psi_i)
	=- \sum_{t=0}^{T-1} \Big(\tilde{f}_{i,t}(\mathbf U_i)+\alpha(t)  p_i(t) \Delta
	 \\ &+ \sum_{k=1}^{n} \underline\xi_{k}(t)  R_{ki}  p_i(t) \Delta - \sum_{k=1}^{n} \overline\xi_{k}(t)  R_{ki}  p_i(t) \Delta \\&-\underline\xi_{i}(t) \underline v_i \Delta  +\overline\xi_{i}(t)  \overline v_i \Delta \Big) - \tilde{\Phi}_i(\mathbf U_i) 
	\\&+  \sum_{t=0}^{T-1} \psi_{i}(t) \Big( p_i(t)\Delta+h_i(\mathbf u_i(t))-a_i(t)\Big).
\end{aligned}
\end{equation}
Considering  \eqref{eq4}, let $(\mathbf U^\star, \mathbf P^\star)$ be an optimal primal solution, and  $(-\boldsymbol{\alpha}^\ast \Delta, \underline{\boldsymbol \Xi}^\ast \Delta, \overline{\boldsymbol\Xi}^\ast \Delta, \boldsymbol{\Psi}^\ast)$ be an optimal dual solution. Since Slater's condition holds, strong duality implies \cite{boyd2004convex}
\begin{equation}\label{eq_primal}
	(\mathbf U^\star, \mathbf P^\star) \in \arg \min_{\boldsymbol U, \boldsymbol P} L(\mathbf U, \mathbf P, \boldsymbol{\alpha}^\ast, \underline{\boldsymbol\Xi}^\ast, \overline{\boldsymbol\Xi}^\ast, \boldsymbol\Psi^\ast),
\end{equation}
for $ \mathbf U_i \in \mathbb{U}_i$. Denote
\begin{equation}
	\begin{aligned}
		\tilde L_i &(\mathbf U_i, \mathbf p_i, \boldsymbol{\alpha}, \underline{\boldsymbol\Xi}, \overline{\boldsymbol\Xi}, \boldsymbol\Psi_i) =  - \sum_{t=0}^{T-1} \Big(\tilde{f}_{i,t}(\mathbf U_i)+\alpha(t)  p_i(t) \Delta
		\\ &+ \sum_{k=1}^{n} \underline\xi_{k}(t)  R_{ki}  p_i(t) \Delta - \sum_{k=1}^{n} \overline\xi_{k}(t)  R_{ki}  p_i(t) \Delta \Big)- \tilde{\Phi}_i(\mathbf U_i) 
		\\&+  \sum_{t=0}^{T-1} \psi_{i}(t) \Big( p_i(t)\Delta+h_i(\mathbf u_i(t))-a_i(t)\Big).
	\end{aligned}
\end{equation}
Considering \eqref{eq6} and \eqref{eq_primal}, we obtain
\begin{equation}
	(\boldsymbol U_i^\star,  \boldsymbol p_i^\star) \in \arg \min_{\boldsymbol U_i, \boldsymbol p_i} L_i( \boldsymbol U_i,  \boldsymbol p_i, \boldsymbol \alpha^\ast, \underline{\boldsymbol \Xi}^\ast, \overline{\boldsymbol \Xi}^\ast, \boldsymbol{\Psi}_i^\ast),
\end{equation}
and, therefore, 
\begin{equation}\label{eq8}
		(\boldsymbol U_i^\star,  \boldsymbol p_i^\star) \in \arg \min_{\boldsymbol U_i, \boldsymbol p_i}  \tilde L_i ( \boldsymbol U_i,  \boldsymbol p_i, \boldsymbol \alpha^\ast, \underline{\boldsymbol \Xi}^\ast, \overline{\boldsymbol \Xi}^\ast, \boldsymbol{\Psi}_i^\ast),
\end{equation}
for $ \mathbf U_i \in \mathbb{U}_i$. Besides, since there holds
\begin{multline}
	(\boldsymbol{\alpha}^\ast, \underline{\boldsymbol \Xi}^\ast, \overline{\boldsymbol\Xi}^\ast, \boldsymbol{\Psi}^\ast) \in \arg \max_{\boldsymbol \alpha, \underline{\boldsymbol \Xi}, \overline{\boldsymbol \Xi}, \boldsymbol{\Psi}} \min_{\mathbf U, \mathbf P} 	L(\mathbf U, \mathbf P, \boldsymbol{\alpha}, \underline{\boldsymbol\Xi}, \overline{\boldsymbol\Xi}, \boldsymbol\Psi) 
	\\= \arg \max_{\boldsymbol \alpha, \underline{\boldsymbol \Xi}, \overline{\boldsymbol \Xi}, \boldsymbol{\Psi}} \min_{\mathbf U, \mathbf P} 	\sum_{i=1}^{n} L_i(\mathbf U_i, \mathbf p_i, \boldsymbol{\alpha}, \underline{\boldsymbol\Xi}, \overline{\boldsymbol\Xi}, \boldsymbol\Psi_i),
\end{multline}
we obtain 
\begin{equation}
	\boldsymbol{\Psi}_i^\ast \in \arg \max_{\boldsymbol \Psi_i} \min_{\mathbf U_i, \mathbf p_i} L_i(\mathbf U_i, \mathbf p_i, \boldsymbol{\alpha}^\ast, \underline{\boldsymbol\Xi}^\ast, \overline{\boldsymbol\Xi}^\ast, \boldsymbol\Psi_i),
\end{equation}
and, therefore, 
\begin{equation}\label{eq7}
	\boldsymbol{\Psi}_i^\ast \in \arg \max_{\boldsymbol \Psi_i} \min_{\mathbf U_i, \mathbf p_i} \tilde L_i(\mathbf U_i, \mathbf p_i, \boldsymbol{\alpha}^\ast, \underline{\boldsymbol\Xi}^\ast, \overline{\boldsymbol\Xi}^\ast, \boldsymbol\Psi_i),
\end{equation}
where the primal and dual variables are considered to be in their respective domains. Additionally, the function $\tilde L_i(\mathbf U_i, \mathbf p_i, \boldsymbol{\alpha}^\ast, \underline{\boldsymbol\Xi}^\ast, \overline{\boldsymbol\Xi}^\ast, \boldsymbol\Psi_i)$ is the Lagrangian  of \eqref{opt_LTD_1} if we define
\begin{equation}
	\lambda_i^\ast(t)=\alpha^\ast(t)+\sum_{k=1}^{n}\underline\xi_k^\ast(t) R_{ki} - \sum_{k=1}^{n}\overline\xi_k^\ast (t) R_{ki}, \quad t \in \mathcal{T}.
\end{equation}
Consequently,  according to \eqref{eq7}, $\mathbf \Psi_i^\ast$ is an optimal dual solution of \eqref{opt_LTD_1}.  Besides, since Slater's condition holds, strong duality implies that according to  \eqref{eq8}, $(\mathbf U^\star, \mathbf P^\star)$ is an optimal primal solution of \eqref{opt_LTD_1}  satisfying \eqref{constraint_balancing}--\eqref{constraint_voltage}. 
Therefore, $(\boldsymbol{\Lambda}^\ast, \mathbf U^\star, \mathbf P^\star)$ is a competitive equilibrium.

(ii) The proof can be obtained by reversing the proof of part (i).
\end{proof}

\begin{proposition}\label{prop_1}
	The locational prices corresponding to a competitive equilibrium satisfy $\lambda_i^\ast(t) \geq 0$ for $i \in \mathcal{N}$, $t \in \mathcal{T}$.
\end{proposition}
\begin{proof}
The proof is straightforward following  \cite[Lemma 1]{salehi2022social}.
\end{proof}

\begin{remark}
	Theorem \ref{theorem1} is an extension of  \cite[Theorem 4]{chen2022competitive} (considering dynamic loads without voltage constraints) and  \cite[Theorem 1]{li2015market} (considering static loads with voltage constraints). The presented proof is related to the previous works, however, it poses new challenges due to the new framework.
\end{remark}
\subsection{Nash Equilibrium}
Arrow and Debreu \cite{arrow1954existence} were the first to show that a competitive equilibrium can be obtained as a solution to a generalized game in which the decision of one player affects both the objective function and the constraint set of other players. In this section, however, we show that the competitive equilibrium in Definition \ref{definition1} can be obtained as a Nash equilibrium of a standard game with $n+1$ players. 
The first $n$ players are prosumers who maximize their payoff subject to the local constraints. Given a price $\lambda_i(t)=\alpha(t)+\sum_{k=1}^{n} \big[\underline\xi_k(t)-\overline\xi_k (t) \big] R_{ki}$ for $t \in \mathcal{T}$, each prosumer $i \in \mathcal{N}$ obtains $(\mathbf U_i, \mathbf p_i)$ as a solution to
\begin{equation}\label{game1}
\begin{aligned}
\max_{{\mathbf U_i}, \mathbf p_i} \quad &  \sum_{t=0}^{T-1} \Big(f_i(\mathbf x_i(t), \mathbf u_i(t)) + \lambda_i(t) p_i(t)\Delta \Big)+{\Phi_i(\mathbf x_i(T)}) \\
{\rm s.t.} 
\quad &  \mathbf x_i(t+1)= \mathbf A_i \mathbf x_i(t)+ \mathbf B_i \mathbf u_i(t),  \\
\quad &  p_i(t)\Delta \leq a_i (t)- h_i(\mathbf u_i(t)),  \\
\quad & \underline{\mathbf x}_i \leq \mathbf x_i(t) \leq \overline{\mathbf x}_i, \\
\quad & \underline{\mathbf u}_i\leq \mathbf u_i(t) \leq \overline{\mathbf u}_i, \quad  t \in \mathcal{T}.\\
\end{aligned}
\end{equation}
The last player is an aggregator who takes the role of a price player and regulates the price such that the power balance constraint and the voltage constraints are satisfied. Given $p_i(t)$ for $i \in \mathcal{N}$ and $t \in \mathcal{T}$, the price player obtains   $(\alpha(t)$, $\underline\xi_i(t),\overline\xi_i(t))$ as a solution to
\begin{equation}\label{game2}
\begin{aligned}
\max_{{\boldsymbol \alpha}, \underline{\boldsymbol \Xi}, \overline{\boldsymbol \Xi}} \quad &  \sum_{i=1}^n \sum_{t=0}^{T-1} \Bigg[-\alpha(t)p_i(t)  
+\underline\xi_{i}(t) \Big( \underline v_i - \sum_{k =1}^n R_{ik}  p_k(t) \Big) 
\\ &+\overline\xi_{i}(t) \Big( \sum_{k =1}^n R_{ik}  p_k(t)-\overline v_i \Big)\Bigg]  \\
{\rm s.t.} 
\quad &  \underline\xi_{i}(t), \overline\xi_{i}(t) \geq 0,
\quad  i \in \mathcal{N}, \,\, t \in \mathcal{T}.
\end{aligned}
\end{equation}
The standard game in \eqref{game1}--\eqref{game2} stems from the generalized game introduced in \cite{arrow1954existence}, where the objective function of the price player indicates the ``law of supply and demand".
\begin{remark}
    It is straightforward to show that a Nash equilibrium $(\mathbf U^\ast, \mathbf P^\ast, \boldsymbol{\alpha}^\ast, \underline{\boldsymbol \Xi}^\ast, \overline{\boldsymbol\Xi}^\ast)$ of the game \eqref{game1}--\eqref{game2} is equivalent to a competitive equilibrium and a social welfare maximization solution.
\end{remark}
\subsection{Strictly Implementable Solutions}
In general, the locational prices are different at each node in the microgrid. In this section, we investigate a special case in which all prosumers have the same locational prices.
\begin{definition}
	Let $(\mathbf U^\star, \mathbf P^\star)$ be a solution to 
	\begin{equation}\label{eq12}
		\begin{aligned}
			\max_{{\mathbf U}, \mathbf P} \quad &  \sum_{i=1}^{n}\Big(\sum_{t=0}^{T-1} f_i(\mathbf x_i(t), \mathbf u_i(t))+ {\Phi_i(\mathbf x_i(T))} \Big)\\
			{\rm s.t.} 
			\quad &  \mathbf x_i(t+1)= \mathbf A_i \mathbf x_i(t)+ \mathbf B_i \mathbf u_i(t),  \\
			\quad &  p_i(t)\Delta \leq a_i (t)- h_i(\mathbf u_i(t)),  \\
			\quad &	\sum_{i=1}^n p_i(t) = 0,\\
			\quad & \underline{\mathbf x}_i \leq \mathbf x_i(t) \leq \overline{\mathbf x}_i, \\
			\quad & \underline{\mathbf u}_i\leq \mathbf u_i(t) \leq \overline{\mathbf u}_i,   \quad  i \in \mathcal{N}, \, t \in \mathcal{T}.
		\end{aligned}
	\end{equation}
	We say $(\mathbf U^\star, \mathbf P^\star)$ is strictly implementable to a microgrid  if it satisfies 
	$$\underline {v}_i < \sum_{k =1}^n R_{ik}  p_k^\star(t)<\overline { v}_i,    \,\,\,  i \in \mathcal{N}, \, t \in \mathcal{T}.$$
\end{definition}

\medskip
\begin{proposition}\label{prop_2}
	If there exists a strictly implementable solution to \eqref{eq12}, then all prosumers have the same locational prices associated with a competitive equilibrium; i.e., $\lambda_i^\ast(t)=\lambda^\ast(t)$ for $i \in \mathcal{N}$, $t \in \mathcal{T}$.
\end{proposition}
\begin{proof}
	According to the proof of part (i) in Theorem \ref{theorem1}, we have
	\begin{equation*}
		\lambda_i^\ast(t)=\alpha^\ast(t)+\sum_{k=1}^{n}\underline\xi_k^\ast(t) R_{ki} - \sum_{k=1}^{n}\overline\xi_k^\ast (t) R_{ki}, \,\,\,  i \in \mathcal{N}, \, t \in \mathcal{T},
	\end{equation*}
	where $\underline\xi_k^\ast(t) \Delta$ and $\overline\xi_k^\ast (t) \Delta$ are the Lagrange multipliers associated with the voltage constraints. Since the solution is strictly implementable, the voltage constraints hold with strict inequality leading to $\underline\xi_k^\ast(t)=\overline\xi_k^\ast (t)=0$ for $k \in \mathcal{N}$, $t \in \mathcal{T}$. Therefore, $\lambda_i^\ast(t)=\alpha^\ast(t)$ for $i \in \mathcal{N}$, $t \in \mathcal{T}$.
\end{proof}

\subsection{Decaying Prices}
The energy price depends on  the available resources and the network topology. In the following, we show that under some assumptions on these two factors, the optimal locational prices decay to  zero after a while. To consider feasible initial conditions, we first define invariant sets.

Suppose $a_i(t)$ is constant over the horizon, i.e., $a_i(t)=a_i$. We denote by $\mathscr{X}=\{ \mathbf x(t)|   \underline{\mathbf x}_i \leq \mathbf x_i(t) \leq \overline{\mathbf x}_i, \text{ for } i \in \mathcal{N}\}$ and $\mathscr{U}=\{ \mathbf u(t)|   \underline{\mathbf u}_i \leq \mathbf u_i(t) \leq \overline{\mathbf u}_i; p_i(t)\Delta \leq a_i-h_i(\mathbf u_i(t)); \sum_{i=1}^{n}p_i(t)=0; \underline {v}_i \leq \sum_{k =1}^n R_{ik}  p_k(t) \leq \overline { v}_i; p_i(t)\in \mathbb{R}, \text{ for } i \in \mathcal{N}\}$  the sets of constraints on the states and control inputs, associated with \eqref{eq4}, respectively. 
\begin{definition}[as in \cite{gutman1987algorithm}]
    Consider the dynamics in \eqref{eq_state} for $i \in \mathcal{N}$. Suppose $(0,0) \in \mathscr{X} \times \mathscr{U}$. A set $ X \subseteq \mathscr{X}$ is called $\mathscr{U}$-invariant w.r.t $\mathscr{X}$ if  for all $\mathbf x(0) \in X$, there exists a control sequence $\{\mathbf u(t)\}_{t=0}^\infty \in \mathscr{U}$ under which $\mathbf x(t) \in X$ for all $t$, and $\lim_{t\rightarrow \infty} \mathbf x(t)=0$. 
\end{definition}
The \textit{maximal $\mathscr{U}$-invariant set} $X_{\text{max}} \subseteq \mathscr{X}$ is the union of all $\mathscr{U}$-invariant sets w.r.t $\mathscr{X}$, which can be approximated by the algorithm presented in \cite{gutman1987algorithm}.

\begin{assumption}\label{assumption2}
	(i) For  $i \in \mathcal{N}$, the functions $f_i(\cdot)$ and $\Phi_i(\cdot)$ are  negative definite (ND) and concave; (ii) the function $h_i(\cdot)$ is convex; (iii) $f_i(\cdot)$, $\Phi_i(\cdot)$, and  $h_i(\cdot)$ go through the origin; (iv) $a_i(t)$ is constant over the horizon, i.e., $a_i(t)=a_i$; (v)  $D:=\sum_{i=1}^{n} a_i>0$; (vi) there holds $\underline{v}_i<\frac{1}{\Delta}\sum_{k=1}^{n}R_{ik}\Big(a_k - \frac{D}{n}\Big) <\overline{v}_i$.
\end{assumption}

\begin{theorem}\label{theorem2}
	Let Assumption \ref{assumption2} hold. If $\mathbf x(0) \in X_{\rm max}$ and $T$ is sufficiently large, then there exists a finite time $\bar T(\mathbf x(0))<T$ such that for $t\geq \bar{T}(\mathbf x(0))$ the locational prices associated with a competitive equilibrium satisfy $\lambda_i^\ast(t)=0$, $i \in \mathcal{N}$.
\end{theorem}
\begin{proof}
	Consider the social welfare maximization problem \eqref{eq4}. Concavity of $f_i(\cdot)$ and $\Phi_i(\cdot)$,  and convexity of $h_i(\cdot)$ result in their continuity. Since $\mathbf x(0) \in X_{\rm max}$, and $f_i(\mathbf x_i(t), \mathbf u_i(t))$ and $\Phi_i(\mathbf x_i(T))$ are negative definite concave functions passing through the origin, if $T \rightarrow \infty$, then $\lim_{t\rightarrow \infty} \mathbf x_i(t)=0$ and $\lim_{t\rightarrow \infty} \mathbf u_i^\ast(t)=0$ for $i \in \mathcal{N}$. Besides, associated with any $\epsilon>0$, there exists a finite time $\bar T(\mathbf x(0))<T$ such that $\|h_i(\mathbf u_i^\ast(t))\|<\epsilon$ and $\sum_{i=1}^{n}h_i(\mathbf u_i^\ast(t))<\sum_{i=1}^{n} a_i$ for $t \geq \bar T(\mathbf x(0))$, $i \in \mathcal{N}$. In the remainder, by $\bar T$ we mean $\bar T(\mathbf x(0))$. 
	Selecting $p_i^\ast(t)$ as
	\begin{equation}\label{eq_e}
		p_i^\ast(t) \Delta = a_i - h_i(\mathbf u_i^\ast(t))  + \frac{1}{n} \left(\sum_{i=1}^{n}h_i(\mathbf u_i^\ast(t))  - \sum_{i=1}^{n}a_i \right), 
	\end{equation}
   the constraints  $p_i^\ast(t)\Delta \leq a_i -h_i(\mathbf u_i^\ast(t))$ and $\sum_{i=1}^{n}p_i^\ast(t)=0$ are satisfied for $t\geq\bar T$. Additionally, if $\epsilon$ is selected to be sufficiently small, there holds 
	$\underline {v}_i< \sum_{k =1}^n R_{ik}  p_k^\ast(t) < \overline { v}_i$,
	for $t \geq \bar T$, $i \in \mathcal{N}$.
	Implementing the first $\bar T$ optimal control inputs $\mathbf u_i^\ast(0), \dots, \mathbf u_i^\ast(\bar T-1)$ on the system dynamics, we arrive at $\mathbf x_i^\ast(\bar T)$ for $i \in \mathcal{N}$. Denote $\mathbf U_{\bar T}^\ast:=(\mathbf u^{\ast \top}(\bar T), \dots, \mathbf u^{\ast \top}(T-1))^\top$.  The principle of optimality implies that $\mathbf U_{\bar T}^\ast$ solves the following constrained optimal control problem 
	\begin{equation}\label{eq9}
		\begin{aligned}
			\max_{{\mathbf U_{\bar T}}} \quad &  \sum_{i=1}^{n}\Big(\sum_{t=\bar T}^{T-1} f_i(\mathbf x_i(t), \mathbf u_i(t))+ {\Phi_i(\mathbf x_i(T))} \Big)\\
			{\rm s.t.} 
			\quad &  \mathbf x_i(t+1)= \mathbf A_i \mathbf x_i(t)+ \mathbf B_i \mathbf u_i(t), \,\,\, \mathbf x_i(\bar T)=\mathbf x_i^\ast(\bar T),\\
			\quad & \underline{\mathbf x}_i \leq \mathbf x_i(t) \leq \overline{\mathbf x}_i, \\
			\quad & \underline{\mathbf u}_i\leq \mathbf u_i(t) \leq \overline{\mathbf u}_i,     \,\,\,  i \in \mathcal{N}, \, t \in \{\bar T, \dots, T-1\}.
		\end{aligned}
	\end{equation}
 As \eqref{eq9} is separable, $\mathbf U_{\bar T}^\ast$ is also a solution to 
	\begin{equation}\label{eq10}
		\begin{aligned}
			\max_{{\mathbf U_{i, \bar T}}} \quad &  \sum_{t=\bar T}^{T-1} f_i(\mathbf x_i(t), \mathbf u_i(t))+ {\Phi_i(\mathbf x_i(T))} \\
			{\rm s.t.} 
			\quad &  \mathbf x_i(t+1)= \mathbf A_i \mathbf x_i(t)+ \mathbf B_i \mathbf u_i(t),  \,\,\, \mathbf x_i(\bar T)=\mathbf x_i^\ast(\bar T),\\
			\quad & \underline{\mathbf x}_i \leq \mathbf x_i(t) \leq \overline{\mathbf x}_i, \\
			\quad & \underline{\mathbf u}_i\leq \mathbf u_i(t) \leq \overline{\mathbf u}_i,     \, \,\, t \in \{\bar T, \dots, T-1\},
		\end{aligned}
	\end{equation}
    where $\mathbf U_{i, \bar T}:= (\mathbf u_i^{ \top}(\bar T), \dots, \mathbf u_i^{ \top}(T-1))^\top$ for $i \in \mathcal{N}$.
	Considering \eqref{eq_e}, we construct $\mathbf P_{\bar T}^\ast=(\mathbf p^{\ast \top}(\bar T), \dots, \mathbf p^{\ast \top}(T-1))^\top$. Following \eqref{eq10}, it can be concluded that $(\mathbf U_{\bar T}^\ast, \mathbf P_{\bar T}^\ast)$ solves the optimization problem in \eqref{opt_LTD_1} starting from $\bar T$ and $\mathbf x_i(\bar T)=\mathbf x_i^\ast(\bar T)$ under $\lambda_i^\ast(t)=0$  such that  \eqref{constraint_balancing}--\eqref{constraint_voltage} are satisfied with the choice of $\alpha^\ast(t)=\underline \xi_{i}^\ast(t) = \overline \xi_{i}^\ast(t)=0$. Therefore, $(\mathbf U_{\bar T}^\ast, \mathbf P_{\bar T}^\ast)$ forms a competitive equilibrium with $\lambda_i^\ast(t)=0$ for $t \geq  \bar T$, $i \in \mathcal{N}$.
\end{proof}

\begin{remark}
    The decaying behavior of price in Theorem \ref{theorem2} implies that with a suitable energy supply, a proper network topology, and good enough initial conditions the available energy will be more than the demand in the long run to stabilize the system. Such assumptions are restrictive in practice since the initial conditions are from a small set and the total supply is always positive. However, a similar case might happen in future energy markets with high penetration of renewables that have zero marginal costs,
    leading to low energy prices.
\end{remark}

\section{Simulation Results}\label{sec:examples}
We consider a radial microgrid as in Fig. \ref{fig0} with $n=9$ nodes and a reference node $0$, where $r_{ij}=0.5$ k$\Omega$ and $x_{ij}=0$ k$\Omega$ for $(i,j) \in \mathcal{E}$. 
\begin{figure}[!t]
	\centering
	\includegraphics[width=2.7 in]{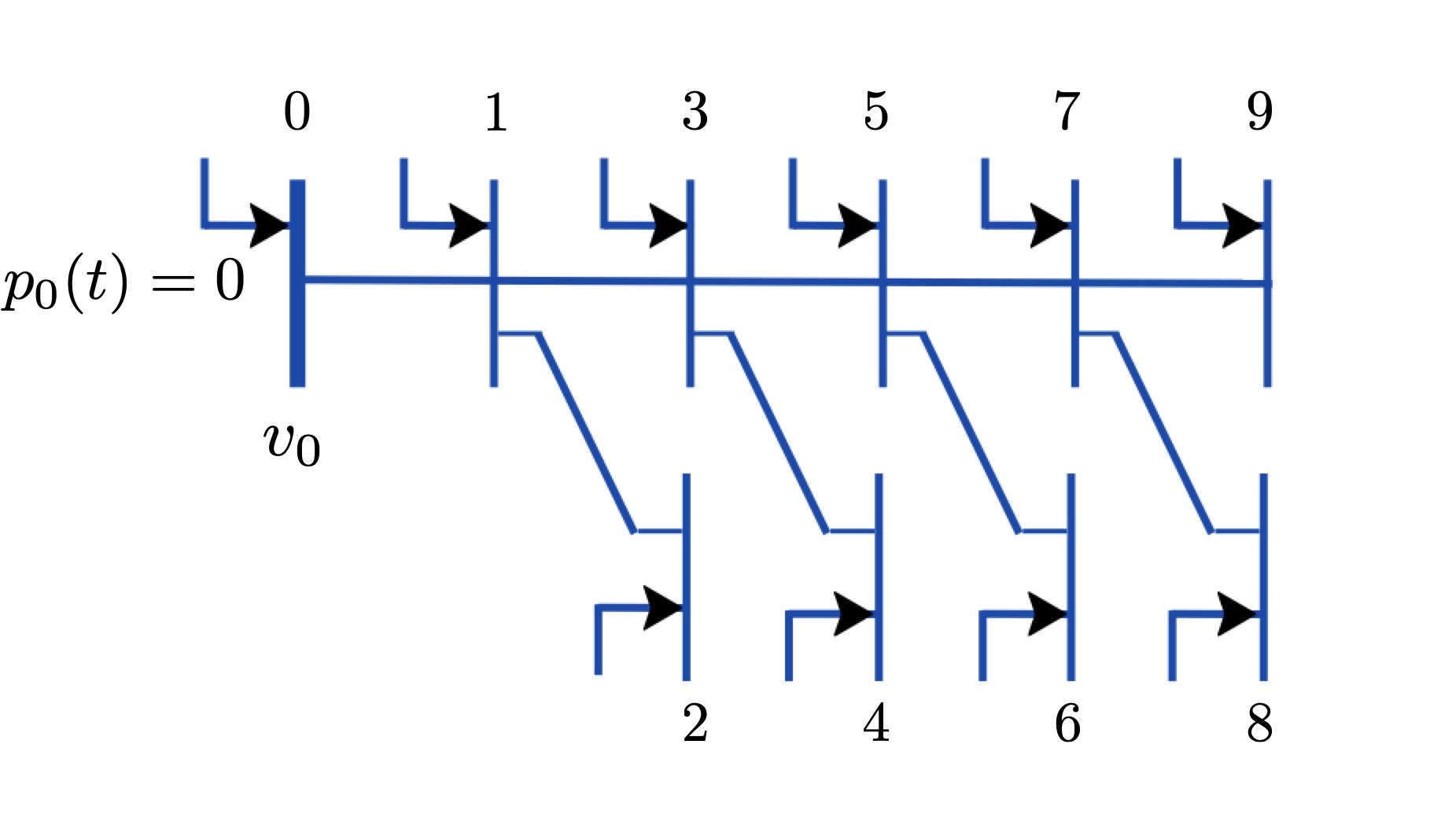}
	\caption{Diagram of the  microgrid in Examples 1 and 2.}
	\label{fig0}
\end{figure}
Let $v_0=12.35^2$ (kV)$^2$, $\overline v_i= (1.05^2-1)v_0$, and $\underline v_i = (0.95^2-1)v_0$ for $i \in \mathcal{N}$. 

\subsection{Real-World Application (EV Charging)}
\begin{example}
Consider the EV charging example  in Section \ref{sec:motivating_example}, where the dynamics are represented by \eqref{eq_EV}.
For  $i \in \mathcal{N}$, let $\eta_i=0.9$ and $\Delta=0.5$. Each EV has $C=30$  kWh battery capacity. The SoC and the charge/discharge rate are constrained by  $0.2C \leq x_i(t) \leq 0.85C$ and $-1.8 \leq u_i(t) \leq 1.8$, respectively. We study the network over a time horizon $T=100$.  Suppose all EVs arrive for charging/discharging at the beginning of the horizon and depart at the end of the horizon. Each prosumer $i \in \mathcal{N}$ tries to maximize its  payoff with the utility functions $f_i(x_i(t), u_i(t))=-u_i^2(t)$ and $\Phi_i(x_i(T))=-(x_i(T)-0.85C)^2$. 

Suppose EVs have  the initial SoC $\mathbf x(0)=(0.23, 0.21, 0.2, 0.24, 0.27, 0.25, 0.3, 0.25, 0.28)^\top C$. The net supply $a_i(t)$ of each prosumer is sinusoidal and depicted in Fig. \ref{fig1}. 
\begin{figure}[!t]
		\centering
		\includegraphics[width=3 in]{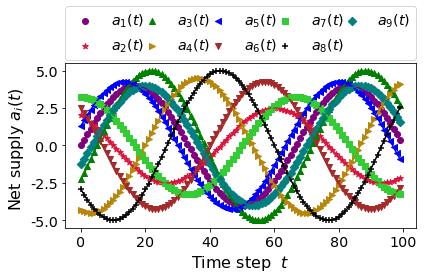}
		\caption{Net supply $a_i(t)$ in Example 1.}
		\label{fig1}
	\end{figure}
We solve the social welfare maximization problem in \eqref{eq4} to obtain the optimal solution $(\mathbf U^\star_i, \mathbf p_i^\star)$, and the optimal locational prices $ \lambda_i^\ast(t)$  as the combination of the Lagrange multipliers associated with the energy balance constraint and the voltage constraints for $i \in \mathcal{N}$. The optimal locational prices are depicted in Fig. \ref{fig3}.
\begin{figure}[!t]
	\centering
	\includegraphics[width=3 in]{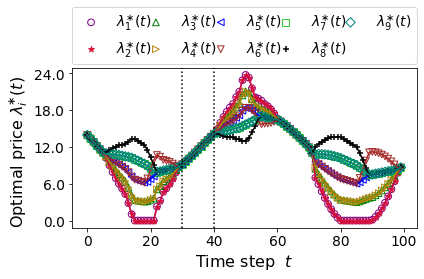}
	\caption{Locational prices  $\lambda_i^\ast(t)$ in Example 1.}
	\label{fig3}
\end{figure}
Considering $ \lambda^\ast_i(t)$, we solve \eqref{opt_LTD_1} to obtain the competitive equilibrium $(\mathbf{U}_i^\ast, \mathbf{p}_i^\ast)$ for $i \in \mathcal{N}$. For illustration,  we depict $(\mathbf U^\star_i, \mathbf p_i^\star)$ and $(\mathbf{U}_i^\ast, \mathbf{p}_i^\ast)$ for $i=1,5,8$ in Fig. \ref{fig4}. 
\begin{figure}[!t]
	\centering
	\includegraphics[width=3.35 in]{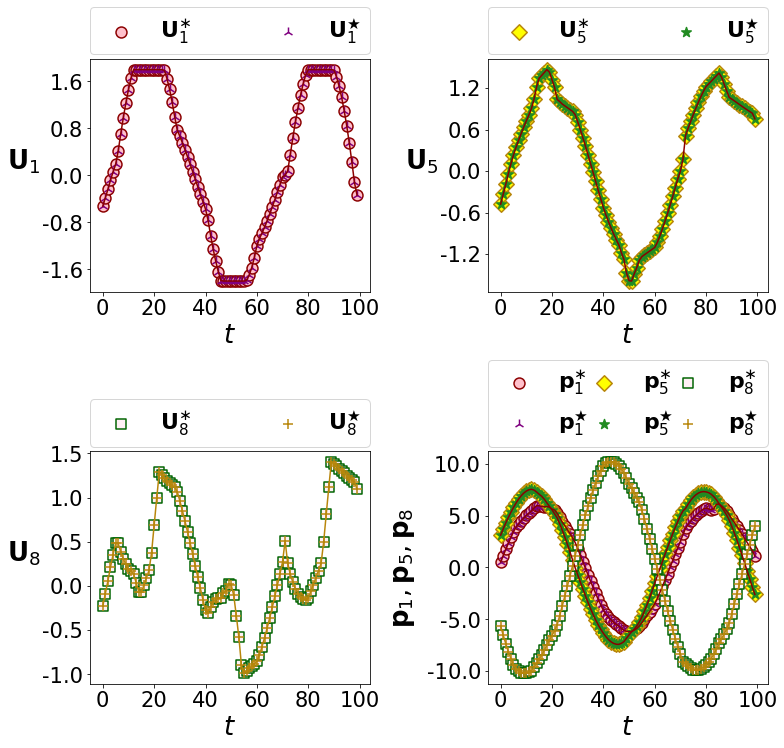}
	\caption{Competitive equilibrium and social welfare maximization solution in Example 1.}
	\label{fig4}
\end{figure}
Additionally, Fig. \ref{fig5} illustrates $\tilde v_i(t):=  \sum_{k =1}^n R_{ik}  p_k^\ast(t)$  for $i \in \mathcal{N}$, where the horizontal dashed lines represent $\underline v_i$ and $\overline{ v}_i$.
\begin{figure}[!t]
	\centering
	\includegraphics[width=3 in]{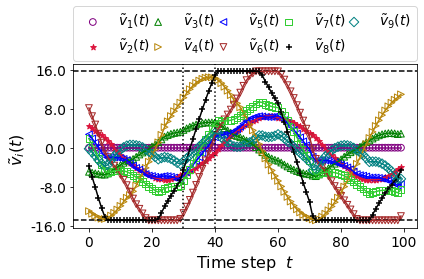}
	\caption{The $\tilde v_i(t)=  \sum_{k =1}^n R_{ik}  p_k^\ast(t)$ in Example 1.}
	\label{fig5}
\end{figure}
Our observations are as follows.
\begin{itemize}
	\item Fig. \ref{fig3} indicates that $\lambda_i^\ast(t)\geq 0$ for $i \in \mathcal{N}$, $t \in \mathcal{T}$, validating Proposition  \ref{prop_1}.
	
	\item Fig. \ref{fig4} illustrates that the competitive equilibrium and the social welfare maximization solution coincide,  validating Theorem \ref{theorem1}.
	
	\item  Considering Fig. \ref{fig5}, when $\underline v_i< \sum_{k =1}^n R_{ik}  p_k^\ast(t) < \overline{ v}_i$ for all $i \in \mathcal{N}$ (e.g., $30 <  t < 40$), we have $\lambda_i^\ast(t)=\lambda^\ast(t)$ as shown in Fig.  \ref{fig3}. This is consistent with Proposition \ref{prop_2}.
\end{itemize}
\end{example}
\subsection{Synthetic Microgrid ($3$-Dimensional Loads)}
\begin{example}
Each node $i\in \mathcal{N}$ is associated with the dynamics in \eqref{eq_state} with the following state-space matrices
\begin{equation}
	\small
	\begin{aligned}
		&{\mathbf A_i} = \left[ {\begin{array}{*{20}{c}}
				{ 1.1}&{ 0}&{0}\\
				{0}&{  0.6}&{0}\\
				{0}&{ 0}&{-0.8}
		\end{array}} \right], \, {\mathbf B_i} = \left[ {\begin{array}{*{20}{c}}
				4&5\\
				2&1\\
				3&5
		\end{array}} \right].
	\end{aligned}
\end{equation}
For $i \in \mathcal{N}$, suppose $\mathbf x_i(0)=(8, 10, 12)^\top$ which lies in the maximal $\mathscr{U}$-invariant set $X_{\text{max}}$. We denote by $\mathbf I$  the identity matrix with an appropriate dimension. Let $\Delta=0.83$ and
\begin{equation}
	\begin{gathered}
		\begin{aligned}	
			f_i(\mathbf x_i(t), \mathbf u_i(t)) &= - \mathbf x_i^\top(t) \mathbf Q_i \mathbf x_i(t)  - \mathbf u_i^\top(t) \mathbf R_i \mathbf u_i(t),\\
			\Phi_i(\mathbf x_i(T)) &= - \mathbf x_i^\top(T) \mathbf Q_i \mathbf x_i(T),\\
			h_i(\mathbf u_i(t)) &=  \mathbf u_i^\top(t)\mathbf H_i \mathbf u_i(t),
		\end{aligned}
	\end{gathered}
\end{equation}
where for $i \in \mathcal{N}$, we have $\mathbf Q_i = 0.5 \mathbf I$, $\mathbf R_i=0.3 \mathbf I$, $\mathbf H_i=3\mathbf I$.
Consider $a_1(t)=a_4(t)=10$, $a_2(t)=-a_9(t)=9$, $a_3(t)=-7$, $a_5(t)=-5$, $a_6(t)=-8$, $a_7(t)=11$, and $a_8(t)=12$. For $i \in \mathcal{N}$, the states and control inputs are bounded by $-12 \times \pmb 1\leq \mathbf x_i(t) \leq 12\times \pmb 1$ and $-10 \times \pmb 1\leq \mathbf u_i(t)\leq 10 \times  \pmb 1$, respectively, where $\pmb 1$ is a vector of an appropriate dimension whose entries are all $1$. Following the same procedure as in Example 1, the locational prices $\lambda_i^\ast(t)$ are obtained by solving \eqref{eq4}, and depicted in  Fig. \ref{fig6}.
\begin{figure}[!t]
	\centering
	\includegraphics[width=3 in]{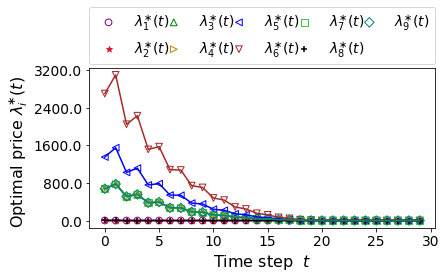}
	\caption{Locational prices $\lambda_i^\ast(t)$ in Example 2.}
	\label{fig6}
\end{figure}
As can be seen, $\lambda_i^\ast(t)=0$ for $t \geq 23$, $i \in \mathcal{N}$, which is consistent with Theorem \ref{theorem2}.
\end{example}
\section{Conclusions}\label{sec:conclusions}
	This paper has considered microgrids with dynamic loads, DERs, and distributed energy allocations. Defining a competitive equilibrium, we presented a local energy market that respects the voltage constraints in the microgrid, inspired by the work presented in \cite{li2015market}. We showed that a competitive equilibrium is equivalent to a Nash equilibrium of a standard game. We proved that under some assumptions,  a competitive equilibrium and a social welfare maximization solution coincide. Considering that the locational prices associated with a competitive equilibrium are different at each node in the microgrid, we investigated a special case leading to the same locational prices. Furthermore, we showed that under some assumptions on the resource supply and network topology, the locational prices decay to zero after a period of time, implying the supply would be more than the demand required to stabilize the system. Finally, we simulated two numerical examples, including EV charging as a real-world application, to validate the results. In future work, extensions to loads with nonlinear dynamics are possible.


\end{document}